\documentclass[a4paper,letter]{IEEEtran}
\usepackage[utf8]{inputenc}
\usepackage[T1]{fontenc}
\usepackage{url}
\usepackage{ifthen}
\usepackage{cite}
\usepackage[cmex10]{amsmath} 

\usepackage{amssymb,amsfonts}
\usepackage{amsthm}
\usepackage{algorithm}
\usepackage{algorithmic}
\usepackage{graphicx}
\usepackage{graphics}
\usepackage{textcomp}
\usepackage[table,xcdraw]{xcolor}
\usepackage{tabularx,booktabs}
\usepackage{pbox} 
\newcommand{\msg}[1]{s_{{#1}}}
\newcommand{\spack}[1]{x_{{#1}}^{(1)}}
\newcommand{\rpack}[1]{x_{{#1}}^{(2)}}
\newcommand{\dpack}[1]{y_{{#1}}^{(2)}}

\newtheorem{theorem}{Theorem}
\newtheorem{lemma}{Lemma}
\newtheorem{proposition}{Proposition}
\newtheorem{definition}{Definition}
\newtheorem{remark}{Remark}

\begin{document}
\title{Subset Adaptive Relaying for Streaming Erasure Codes}

\author{%
  \IEEEauthorblockN{Muhammad Ahmad Kaleem, Gustavo Kasper Facenda and Ashish Khisti\\}
  \IEEEauthorblockA{Department of Electrical and Computer Engineering\\ University of Toronto \\ 
                    Toronto, ON M5S 1A1, Canada\\
  }
}

 \maketitle

\begin{abstract}

  This paper investigates adaptive streaming codes over a three-node relayed network. In this setting, a source transmits a sequence of message packets through a relay under a delay constraint of $T$ time slots per packet. The source-to-relay and relay-to-destination links are unreliable and introduce a maximum of $N_1$ and $N_2$ packet erasures respectively.  Recent work has proposed adaptive (time variant) and nonadaptive (time invariant) code constructions for this setting and has shown that adaptive codes can achieve higher rates. However, the adaptive construction deals with many possibilities, leading to an impractical code with very large block lengths. In this work, we propose a simplified adaptive code construction which greatly improves the practicality of the code, with only a small cost to the achievable rates. We analyze the construction in terms of the achievable rates and field size requirements, and perform numerical simulations over statistical channels to estimate packet loss probabilities. 
  \end{abstract}

  \section{Introduction}

  Many modern applications including online gaming and video conferencing require efficient low-latency communication. In these applications, data packets are generated sequentially at the source and must be recovered under strict delay constraints at the destination. When packets are lost over the network, significant amounts of error propagation can occur and thus suitable methods for error correction are necessary. 
  
  There are two main approaches for error correction against packet losses in communication networks: Automatic repeat request (ARQ) and Forward error correction (FEC). ARQ, which involves retransmission, is not suitable when considering low latency constraints over long distances. Therefore, FEC schemes are considered more appropriate for low latency settings. FEC codes with strict decoding delay constraints have been specifically referred to as streaming codes. 
  
  While most prior work on streaming codes has focused on a point-to-point communication link, a network topology of practical interest is a three-node relay network which involves a relay node between the source and destination. This topology is motivated by numerous applications in which a gateway server connects two end nodes. Streaming codes in this setting were first introduced in \cite{silas2019} and the time-invariant capacity was derived. Following this, \cite{adaptiverelay2023} proposed adaptive code constructions where the relay uses different rate codes depending on the number of erasures observed in the source-to-relay link. This leads to more efficient relay-to-destination codes and allows the overall rate to be improved whenever the relay-destination link is the bottleneck.

  On the other hand, the scheme presented in \cite{adaptiverelay2023} requires prohibitively large packet sizes in most cases. For example, for a choice of parameters, that scheme would require 100~KB per packet, which is larger than the maximum packet size for UDP, and significantly larger than any practical packet size considered. This stems from the use of long maximum distance separable (MDS) codes, which lead to a high field size, as well as large packet sizes in terms of symbols (i.e., number of elements of said field present in a packet). In order to obtain a practical construction, we propose an adaptation scheme that employs short MDS codes, thus reducing the field size, and having smaller overall packet sizes in terms of number of symbols. In order to do so, we reduce the adaptation power of the relay, which slightly reduces the resulting rate, but allows the coding scheme to work with packets as small as 60 bytes, compared to the previous 100~KB. We note that these changes also make the relaying and decoding computationally simpler, although this is not a focus of the paper.

  We then evaluate our scheme under statistical erasure models, simulating the resulting packet loss rates. This was not done previously, as the complexity of the previous scheme made simulations intractable.

  Finally, we show that our construction can be naturally extended to the multi-user setting \cite{kasper2022}, and, by doing so, we are able to achieve rates higher than the sumrate upper bound for non-adaptive schemes presented in that work.

  The rest of the paper is organized as follows. Section \ref{sec:background} discusses background and prior work on streaming codes in different contexts and section \ref{sec:sysmodel} presents the system model for the three node relay network. In Section \ref{sec:mainres}, we introduce the main result of this work which is followed by the details of the proposed coding scheme in section \ref{sec:subset} and the proof of the result in section \ref{sec:anal}. Finally, section \ref{sec:numresult} shows results of numerical simulations of the proposed scheme while section \ref{sec:extension} presents the extension to the multiaccess network setting. 

  \section{Background and Prior Work}\label{sec:background}

  Prior work has studied different types of streaming codes to establish fundamental limits of reliable low-latency communication under different packet-loss models. In particular, \cite{martinian2004burst} has studied the point-to-point network (one source node, one destination node) under burst erasure sequences, \cite{leong2012erasure} has studied burst and arbitrary erasures seperately, and \cite{badr2013streaming} has extended the erasure patterns, allowing for both burst and arbitrary erasures. Other works that have further studied different aspects of streaming codes include \cite{joshi2012playback,Karzand2017,badr2017layered,badr2017fec,krishnan2018rate,fong2019optimal,domanovitz2019explicit,KrishnanLowField2020,Haghifam2021}.
  
  The three node relay network setting we focus on in this paper can be used to model common real world communication scenarios such as video conferencing and online gaming where two users communicate through an intermediate server. In these settings, low latency is desired as it can greatly improve user experience \cite{Jarschel2013,Quax2013,Clincy2013,Claypool2014,Slivar2014}. The latency itself may come from various sources, including hardware delay, propagation delay, server-side delay, or communications delay. Optimizations to reduce several of these delays have been studied e.g. \cite{lee2014outatime,Slivar2015}, however the communications delay has only recently been explored by \cite{adaptiverelay2023}. Since the communications delay is a signficant part of the overall latency, retransmissions would represent a significant cost and thus streaming codes can help reduce the overall delay significantly. 
  
  Streaming codes in the three node relay setting were first studied by \cite{silas2019} where the time-invariant capacity was derived and it was shown that symbol-wise decoding methods, as opposed to traditional message-wise methods, are optimal. \cite{adaptiverelay2023} extended these results with an adaptive relaying strategy in which the relay takes into account the erasure pattern from source to relay when forwarding symbols to the destination. This method led to strictly higher achievable rates than \cite{silas2019}.
  Finally, \cite{kasper2022} has studied streaming codes in the multiaccess network setting, where the network is modeled similarly to the three node relay network but with multiple source nodes. There, the time-invariant capacity region was derived by extending ideas from time-invariant single user codes.

  \section{System Model}\label{sec:sysmodel}

  In this section, we formally introduce the problem setting and relevant notation. We denote the set of integers by $\mathbb{Z}_{+}$, the finite field over $q$ elements by $\mathbb{F}_q$ and the set of $l$-dimensional column vectors over $\mathbb{F}_q$ by $\mathbb{F}_q^l$. For $a, b \in \mathbb{Z}_{+}$, we use $[a:b]$ to denote $\{i \in \mathbb{Z}_{+} | a \leq i \leq b\}$. Similarly, we let $[a:c:b] = \{i \in \mathbb{Z}_{+} | i = a + kc \;, 0 \leq k \leq \lfloor \frac{b-a}{c} \rfloor \}$.

    We consider a network with one source, one relay and one destination. The source wishes to transmit a sequence of messages $\{s_{t} \}_{t = 0}^{\infty}$ to the destination through the relay. We assume there is no direct link between the source and destination, and that the packet communication in both links is instantaneous i.e. no propagation delay. We assume that the link between the first source and the relay introduces at most $N_1$ erasures, and that the link between the relay and the destination introduces at most $N_2$ erasures. 
    The destination wishes to decode the source packet with a maximum delay of $T$ timeslots. A visual representation of this network is shown in Figure \ref{fig:3node}.

  \begin{figure}[htbp]
    \centerline{\includegraphics[width=\linewidth]{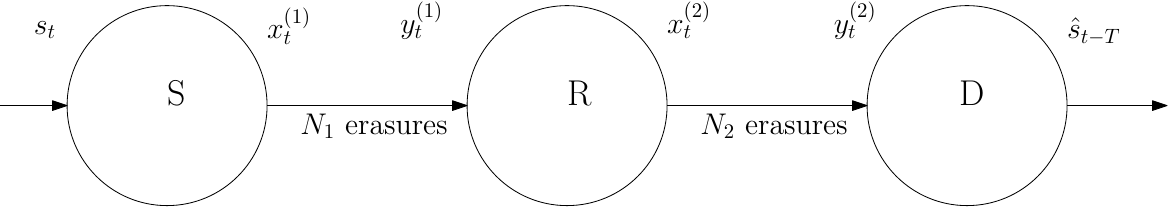}}
    \caption{Three node relay network}
    \label{fig:3node}
    \end{figure}

The following definitions from \cite{silas2019} formalize the notion of a streaming code in the three node setting and the associated concepts of erasures and achievable rates.

\begin{definition}\label{def:streamingcode}
      $(n_1, n_2, k, T)_{q}$-streaming code consists of the following:
      \begin{itemize}
          \item A sequence of source messages $\{s_{t}\}_{t=0}^{t = \infty}$ where $s_{t} \in \mathbb{F}_q^{k}$.
          \item An encoding function $$f_{t} : \underbrace{\mathbb{F}_q^{k} \times \cdots \times \mathbb{F}_q^{k} }_{t+1 \textrm{ times}} \to \mathbb{F}_q^{n_1}$$
      used by the source at time $t$ to generate $\spack{t} = f_{t} (\msg{0}, \msg{1}, \dots, \msg{t})$.
      \item A relaying function

         \begin{equation*}
              g_t: \underbrace{\mathbb{F}_q^{n_1} \cup \{*\} \times \dots \times \mathbb{F}_q^{n_1} \cup \{*\}}_{t+1 \; \text{times}} \to \mathbb{F}_q^{n_2}   
          \end{equation*}

      used by the relay at time $t$ to generate 
          
          \begin{equation*}
              x_t^{(2)} = g_t(\{y_{i}^{(1)}\}_{i=0}^t )
          \end{equation*}
  
      \item A decoding function 
      \begin{align*}
          \varphi_{t} = \underbrace{\mathbb{F}_q^{n_2} \cup \{*\} \times \cdots \times \mathbb{F}_q^{n_2} \cup \{*\}}_{t+T+1\textrm{ times}} \to \mathbb{F}_q^{k} \\
      \end{align*}
      used by the destination at time $t + T$ to generate
      \begin{align*}
          \hat{s}_{t} = \varphi_{t}(y^{(2)}_{0},y^{(2)}_{1},\ldots,y^{(2)}_{t+T})\\
      \end{align*}
      \end{itemize}
  \end{definition}

  \begin{definition}\label{def:erasureseq}
    An erasure sequence is a binary sequence denoted by $e^{(1)} \triangleq \{e^{(1)}_{t}\}_{t = 0}^{\infty}$, where $e_t^{(1)}$ denotes whether an erasure occurs in the source-to-relay link at time $t$.

      Similarly, $e^{(2)} \triangleq \{e^{(2)}_t\}_{t = 0}^{\infty}$ where $e_t^{(2)}$ denotes whether an erasure occurs in the relay-to-destination link at time $t$.

    An $N$-erasure sequence is an erasure sequence $e$ that satisfies $\sum_{t=0}^{\infty} e_t = N$. In other words, an $N$-erasure sequence specifies $N$ arbitrary erasures on the discrete timeline. The set of $N$-erasure sequences is denoted by $\Omega_N$. 
  \end{definition}
  
  \begin{definition} \label{def:channel}
    The mapping $h_n : \mathbb{F}_q^{n} \times \{0, 1\} \to \mathbb{F}_q^{n} \cup \{*\}$ of an erasure channel is defined as
    \begin{align}
    h_n(x, e) = \begin{cases}
      x, &\textrm{~if~} e = 0\\
    *, &\textrm{~if~} e = 1  
    \end{cases} \label{eq:erasuremodel}
    \end{align}
    For any erasure sequence $e^{(1)}$ and any $(n_1, n_2, k,  T)_{q}$-streaming code, the following input-output relation holds for each $t \in \mathbb{Z}_+$:
    \begin{align}
    y_{t}^{(1)} = h_{n_1}(x^{(1)}_{t}, e^{(1)}_{t})\\
    \end{align}
    where $e^{(1)} \in \Omega_{N_1}$.
    Similarly, the following input-output relation holds for for each $t \in \mathbb{Z}_+$:
    \begin{align}
    y_{t}^{(2)} = h_{n_2}(x_{t}^{(2)}, e^{(2)}_{t}) \label{eq:erasuredest}
    \end{align}
    where $e^{(2)} \in \Omega_{N_2}$.
  \end{definition}

  \begin{definition}
    An $(n_1, n_2, k,  T)_{q}$-streaming code is said to be $(N_1, N_2)$-achievable if, for any $e^{(1)} \in \Omega_{N_1}$ and $e^{(2)} \in \Omega_{N_2}$, for all $t \in \mathbb{Z}_+$ and all $s_{t} \in \mathbb{F}_q^{k}$, we have $\hat{s}_{t} = s_{t}$.
  \end{definition}

\begin{remark}\label{rem:1}
    While Definition \ref{def:erasureseq} may appear limiting as it considers $\Omega_{N_1}, \Omega_{N_2}$ having a total of $N_1, N_2$ erasures in $[0:\infty]$ respectively, this condition can be improved to sequences with $e^{(1)}, e^{(2)}$ satisfying $\sum_{i'=i}^{i+T} e_{i'}^{(1)} \leq N_1$ and $\sum_{i'=i}^{i+T} e_{i'}^{(2)} \leq N_2$ (i.e. in any sliding window of $T+1$ consecutive time slots, the source-relay and relay-destination links see at most $N_1, N_2$ erasures respectively) by noting that messages have a delay constraint of $T$. 
\end{remark}

  \begin{definition}
    The rate of an $(n_1, n_2, k,  T)_{\mathbb{F}}$-streaming code is \begin{align*}
        R &= \frac{k}{\max(n_1, n_2)} \\
    \end{align*}
  \end{definition}

\section{Main Results}\label{sec:mainres}

The following theorem is the main result of this work. 

\begin{theorem}\label{thm:rates}
  For any $N_1, N_2$ and $T$, there exists an $(N_1, N_2)$-achievable streaming code with rate $R = \min(R_1, R_2)$ for 
  \begin{align}
      R_1 &= \frac{T+1-N_1-N_2}{T+1-N_2} \\
      R_2 &= \frac{T+1-N_2-j}{T+1-N_1 + (N_1 - j) \cdot \frac{T+1-N_2-j}{T+1-N_2-N_1} + \delta} 
  \end{align} 

  Furthermore, the minimum field size $q$ required for this code is $q = T+1-j$. 
\end{theorem}

\begin{remark}
  The $\delta$ term in Theorem \ref{thm:rates} represents an overhead to inform the destination of the erasure pattern which occurs from source to relay. In order to keep the construction simple, we assume that the destination already has access to the erasure pattern. We present details to compute $\delta$ in Section \ref{sec:achievable}.
\end{remark}

The code construction that achieves this rate is presented in Section~\ref{sec:subset}. We then prove the theorem in three parts in Section~\ref{sec:anal}. First, we show that the code constructed by our algorithm has rate $R$. Then, we show that a field size of $q = T + 1 - j$ is sufficient to encode the packets with MDS codes. Finally, we show that this code is $(N_1, N_2)$-achievable.

This result provides some key observations that will be discussed in more detail in the upcoming sections. The first is that, comparing our results to \cite{silas2019} and \cite{adaptiverelay2023}, we find 
that most of the rate improvement can be obtained with little adaptation power from the relay, with significantly diminishing gains obtained from more adaptation. On the other hand, the cost in further increasing this adaptation power is not negligible as we will see in Section~\ref{sec:numresult}, 
thus our scheme presents a good trade-off point. We also note that the choice of which erasure pattern to adapt to, which is described by the choice of $j$ in our scheme, influences the resulting rate, with little impact on complexity, making an optimal choice desirable. Finally, the field size aspect of the result is a significant novelty of our work and requires using a fundamentally different idea than the adaptive code from \cite{adaptiverelay2023}. Specifically, we introduce a concept of grouping for encoding of the relay to destination packets which ensures that concatenations of short MDS codes are used instead of a single long MDS code as was the case in prior work. Since the field size requirement depends on the underlying code, which is now much shorter than before, we are then able to achieve the claimed amount in Theorem \ref{thm:rates}. This gain we achieve is also significant as the value of $q$ required in \cite{adaptiverelay2023} is of the order $\mathcal{O} ((T+1-N_1)^{N_1})$. 

  \section{Proposed Coding Scheme}\label{sec:subset}
  
  In this section, we present our proposed adaptive relaying scheme. 
  One key idea behind our scheme is that, rather than adapting to all possible erasure patterns that may occur in the source-to-relay link, the relay only adapts to a small subset of erasure patterns. Specifically, by selecting a value $0 \leq j < N_1$, for
  any given time $t$ and source message $s_t$, if the number of erasures observed 
  in the time window $[t : t + T - N_2]$ is less than or equal to $j$, the relay will start transmitting symbols from $s_t$ at time $t+j$ with the adaptive rate 
  and otherwise will start transmitting symbols at time $t+N_1$ with the nonadaptive rate. 
  As mentioned previously, the motivation behind adapting to a subset of the possible erasure patterns is to offer a practical code that still enjoys most of the benefits of adaptation.

  \subsection{Source-to-relay encoding}  
    
  The construction of the source-to-relay code involves an $(n', k')$ diagonally interleaved MDS code with $n' = T+1-N_2$ and $k' = T+1-N_1-N_2$ concatenated $l' = T+1-N_2-j$ times. The concatenations can be seen as multiple "layers" of the same code. We then obtain the following code parameters:
  \begin{align}
      k &= (T+1-N_1-N_2)(T+1-N_2-j) \\
      n_1 &= (T+1-N_2-j)(T+1-N_2),
  \end{align}
  
and a rate $R_1 = \frac{k}{n_1} = \frac{T+1-N_2-N_1}{T+1-N_2}$. The idea behind the use of multiple concatenations is to have $k_1$ be divisible by $T+1-N_2-N_1$ and $T+1-N_2-j$ which we will show is required later on for the relay to destination encoding. We now introduce the concept of symbol \textit{estimates} and an associated proposition from \cite{adaptiverelay2023}.
  
      \begin{definition}
          We say $\tilde{s}_i[l] \in \mathbb{F}_q$ is an estimate of a source symbol $s_i[l]$ if there exists a function $\Psi_{i, l}$ such that $\Psi_{i, l} (\tilde{s}_i[l], \{s_t\}_{t \in [0:i-1]}) = s_i[l]$. 
      \end{definition}

  \begin{proposition}[Proposition 1, \cite{adaptiverelay2023}] \label{prop:indices}
      Assume that packet $x_t^{(1)}$ is erased. Then, denote by $\mathcal{I} = \{t_1, \dots, t_{T+1-N_1-N_2} \}$ the (ordered) time indices of the first $T+1-N_1-N_2$ non-erased source encoded 
      packets after time $t$, and denote by $t_v$ the $v$th element of the
      set. Then, at time instant $t_v$, the relay has access to a set of
      estimates $\tilde{\mathcal{M}}^{v}$ for which the following properties hold:
      \begin{enumerate}
          \item $|\tilde{\mathcal{M}}^{v}| = l' v$
          \item $H(s_t | \tilde{\mathcal{M}}^{v}, \{s_i\}_{i=0}^{t-1}) \leq k - l'v$
      \end{enumerate}
  
  \end{proposition}
  
  Based on this proposition, we know that from each nonerased packet $x_{t+i}^{(1)}$ ($1 \leq i \leq n'-1$) (up to $k'$ such packets in total), the relay obtains $l'$ symbol estimates of $s_t$. Moreover, if $x_t^{(1)}$ is not erased, the relay obtains all symbols from $s_t$ at time $t$ since the code is systematic.

  \subsection{Relay-to-destination encoding}\label{sec:subsetrdenc}
  
  For the relay-to-destination link, we present our adaptive code construction based on the use of symbol estimates mentioned in the previous section. We have two possibilities for the encoding depending on whether $x_t^{(1)}$ is erased. 
  
  \subsubsection{$x_t^{(1)}$ is not erased}
  
  If $x_t^{(1)}$ is not erased, all of the symbols from $s_t$ are available at time $t$ and the relay uses a $(n'', k'')$ diagonally interleaved MDS code with $k'' = T+1-N_2-j$ and $n'' = T+1-j$. This code is concatenated $l'' = \frac{k}{k''} = T+1-N_1-N_2$ times. The diagonally interleaved code is spread out across the relay packets $\rpack{t+j}, \dots, \rpack{t+T}$.

  \subsubsection{$x_t^{(1)}$ is erased}
  
  In this case, the relay does not receive all symbols from $s_t$ at time $t$ and the relay encoding will depend on the number of erasures which happen following time $t$. We start by defining some additional terminology. We will let $\mathbf{C}(t) = [C(t; 1) \; C(t; 2) \; \dots \; C(t; T)]$ be the symbols, originally coming from $s_t$, that are transmitted by the relay in the time window $[t+1:t+T]$. 
      Here, the $C(t; i)$ will then consist of symbols sent at time $t+i$ (to be seen as a row vector) and these will naturally be a function of the nonerased source packets in the interval $[0:t+i]$ which contain symbols from $s_t$.
  We further let $\alpha_t(t+i) = |C(t; i)|$. Based on the source-to-relay encoding, we always have $0 \leq \alpha_t(t+i) \leq l'$. Similarly, we denote the cumulative number of symbol \textit{estimates} of code symbols from $s_t$ available to the relay at time $t+i$ by $\kappa_t(t+i)$. This value can be calculated exactly since each nonerased timeslot leads to an additional $T+1-N_2-j$ symbol estimates being available. Therefore, multiplying by the number of nonerased times from $t$ to $t+i$ and taking the maximum with $k$ gives the desired value of $\kappa_t(t+i)$. Lastly, we let the cumulative number of erasures in the interval $[t+1: t+i-1]$ be represented by $\gamma(i)$. 
  
  The coding strategy is then based on including $C(t;i)$ as a subpacket in the relay packet $x_{t+i}^{(2)}$. Here, only values of $i$ in $[j:T]$ are relevant as the earliest the relay starts transmitting is at time $t+j$. Therefore by default, for $i \in [0, j-1]$, we let $C(t; i)$ be an empty vector so that $\alpha_t(t+i) = 0$. Moreover, from $i = T-N_2 + 1$ to $i  =T$, $C(t; i)$ will consist of parity symbols to ensure that the code is robust to any $N_2$ erasures which may occur from relay to destination. We now present Algorithm \ref{alg:subset} which provides a way to calculate the number of symbols to include at time $t+i$, i.e. $\alpha_t(t+i)$, for $i \in [j, T-N_2]$.

  \begin{algorithm}[H]
      \caption{Computation of $\alpha_t(t+i)$ for $i \in [j:T-N_2]$}\label{alg:subset}
      \begin{algorithmic}[1]
      \STATE $i \gets j$
          \WHILE {$i \leq T - N_2$}
          \STATE $\gamma(i) \gets \text{number of erasures in time slots} \; [t+1:t+i-1]$
          \STATE $\kappa_t(t+i) \gets \max(k, (T+1-N_2-j) \cdot (i-1-\gamma(i)))$
          \IF {($\gamma(i) \leq j-1$)}
          \STATE $\ell_i \gets T+1-N_2-N_1$
          \ELSIF {($\gamma(i) \geq j$ \textbf{and} $i \geq N_1$)}
          \STATE $\ell_i \gets T+1-N_2-j$ 
          \ELSE 
          \STATE $\ell_i \gets 0$
          \ENDIF
          \STATE available $\gets \kappa_t(t+i) - \sum_{a \in [0:i-1]} \alpha_t(t+a) $
          \STATE $\alpha_t(t+i) \gets \min \{\ell_i, \text{available} \}$
          \STATE $i \gets i +1$
          \ENDWHILE
      \end{algorithmic}
      \end{algorithm}

We note here that because there are only two possible rates of transmission, we have $\ell_0 = \dots = \ell_{j-1} = \frac{k}{T+1-N_2-j} = T+1-N_1-N_2$ and $\ell_j = \dots = \ell_{N_1-1} = \frac{k}{T+1-N_2-N_1} = T+1-N_2-j$. 

\subsubsection{Construction of Parities}\label{sec:paritydesc}

We now present the details behind how the parity subpackets i.e. $C(t; i)$ for $T-N_2 + 1 \leq i \leq T$ are constructed. The main idea is to have the relay construct the parities according to a $(T+1-j, T+1-j-N_2)$ or a $(T+1-N_1, T+1-N_1-N_2)$ diagonal MDS code depending on if $\gamma(T-N_2)$ is $\leq j-1$ or $\geq j$ respectively. 
 We therefore treat the overall code as a diagonal MDS code with these parameters even if the $C(t; i)$ for $i \leq T-N_2$ are not of a consistent size.  

In the case where $x_t^{(1)}$ is not erased, all of the $C(t; i)$'s have a constant size and so the construction of the parity subpackets follows trivially from the standard construction of MDS codes which involve linearly independent combinations of the code symbols. Specifically, a $(n'', k'')$ MDS code and its concatenations will suffice. Similarly, if $j = 0$ and $x_t^{(1)}$ is erased, the relay only starts to transmit symbols from time $i = t + N_1$ onwards from where the size of the $C(t; i)$'s remains constant and thus the parity subpackets can once again be constructed in a standard way.

The details which follow are thus focused on the case where $x_t^{(1)}$ is erased and $ j > 0$. In this case the sizes of the $C(t; i)$ may change and be smaller than $\ell_{\gamma(i)}$ if not enough symbols are available, as shown in Algorithm \ref{alg:subset}. \cite{adaptiverelay2023} dealt with this case by using long MDS codes but here we will show that concatenations of a $(T+1-j, T+1-j-N_2)$ or $(T+1-N_1, T+1-N_1-N_2)$ code are sufficient. This is another advantage of our approach in that it allows for a low field size to be used.

 If we have $\gamma(T-N_2) \leq j - 1$, we know that the relay will always have had enough symbols to transmit since for any $i$, at time $t+i$ the relay has transmitted $i - j + 1$ packets but has received at least $i+1 - (j-1) \geq i-j+1$ packets. Thus, using a $(T+1-j, T+1-j-N_2)$ code here also follows directly as all of the $C(t; i)$'s have equal size and the parity symbols can be generated directly as linear combinations of these. 

  However, if we have $\gamma(T-N_2) \geq j$ and the relay had started transmitting at time $t +j$, then it must have observed erasures afterwards causing $\gamma(i) = j$ for some $i$ i.e. a switch in the rate used by the relay. This process can be referred to as a within-message variable rate. To formalize how the relay proceeds in such a scenario, we define two relaying subfunctions, $g_{t+i, \text{adapt}}$ and $g_{t+i, \text{nonadapt}}$ which the overall relaying function $g_{t+i}$ (Definition \ref{def:streamingcode}) is composed of. 
  
  \begin{align*}
      g_{t+i, \text{adapt}}: \mathbb{F}_q^{n_1} \cup \{*\} \times \dots \times \mathbb{F}_q^{n_1}  \cup \{*\} \to \mathbb{F}_q^{\ell_0} \\
      g_{t+i, \text{nonadapt}}: \mathbb{F}_q^{n_1} \cup \{*\} \times \dots \times \mathbb{F}_q^{n_1} \cup \{*\} \to \mathbb{F}_q^{\ell_j} \\
  \end{align*}

  These subfunctions represent how the symbol estimates of time $t$ are used to generate the corresponding contribution to relay packets $\rpack{t+i}$, in the form of its subpacket which consists of the symbols from time $t$.  The adaptive subfunction $g_{t+i, \text{adapt}}$ is used if $\gamma(i) \leq j-1$ and the nonadaptive subfunction is used otherwise. 
  
  We now consider $s_t$ and note that for each nonerased source to relay packet $x_t^{(1)}, \dots, x_{t+T-N_2}^{(1)}$, the relay receives $l' = \frac{k}{T+1-N_1-N_2} = T+1-j-N_2$ symbol estimates of $s_t$ (except for time $t$ where it receives all $k$ symbols). Now assume that the relay begins transmitting with the rate corresponding to $j$ erasures at time $t+j$. At this rate, the relay will transmit $\ell_0 = T+1-N_1-N_2$ symbols per time slot and the adaptive relaying function $g_{t+i, \text{adapt}}$ is used. Then assume that at time $t+j+i$, $i \geq 1$, the relay observes its $j+1$st erasure from source to relay and thus changes its rate of transmission to the one corresponding to $N_1$ erasures. Here, $\ell_{j} = T+1-j-N_2$ symbols are sent per timeslot and now the nonadaptive relaying function is used. Now the relay will wait until time $t + N_1$ if $i < N_1$ and from $t+N_1$ onwards, start transmitting $\min(T+1-j-N_2, \textrm{available})$ symbols per timeslot. This will continue until time $t+T-N_2$. Following this, it will generate parity symbols to send until time $t+T$. 
  
 To generate the parity symbols, the relay will continue to use the nonadaptive subfunction. Specifically, it will first divide source symbol estimates from time $t$ into $T+1-N_1-N_2$ groups, each of size $T+1-j-N_2$. Let these groups be denoted by $g_{r, t} = \tilde{s}_t[1+(r-1)(T+1-j-N_2):r(T+1-j-N_2)]$ for $1 \leq r \leq T+1-N_1-N_2$. We further assume WLOG that the estimates $\tilde{s}_t[p]$ are arranged according to the order in which they are transmitted initially by the relay i.e. $\tilde{s}_t[p]$ is transmitted at an earlier or equal time as $\tilde{s}_t[q]$ if and only if $p < q$. It then follows that the groups $g_r$ are also arranged by the times at which the relay transmits the symbols within. The parity subpackets are then generated based on linear combinations of the $g_r$ i.e. assuming the previous relay subpackets were the $g_r$'s and were sent nonadaptively. Overall, this corresponds to using a $(T+1-N_1, T+1-N_1-N_2)$ MDS code concatenated $|g_{r, t}| = T+1-j-N_2$ times. In other words, for each value of $p \in [1:T+1-j-N_2]$, the relay uses a $(T+1-N_1, T+1-N_1-N_2)$ MDS code with $g_{1, t}[p], \dots, g_{T+1-N_1-N_2, t}[p]$ as the message symbols. We prove the correctness of this idea in Section \ref{sec:fieldsize}.
  
  \subsection{Examples}

  In this section, we present two detailed examples to illustrate the main concepts involved in our scheme. 
  
  For the first, we let the system parameters be $T = 5, N_1 = 2$ and $N_2 = 3$. Then with $j = 0$, we have $k' = T+1-N_1-N_2 = 1$ and $n' = T+1-N_2 = 3$ with $l ' = T+1-N_2 - j = 3$ concatenations giving $k = 3$ and $n_1 = 9$. The source-to-relay code is shown in Table \ref{tab:srex1}. We represent the concatenations of the code by having the source symbols at time $i$ be $s_i[1], s_i[2], s_i[3]$. Each of these symbols is present in a diagonal $(3,1)$ MDS code. 
  
  We then assume that there is a burst erasure of length 2 at times 1 and 2. Based on this erasure pattern, the relay to destination code will be as in Table \ref{tab:rdex1}. Since there are no erasures at times 0, 3, 4, 5, symbols $s_0, s_3, s_4, s_5$ are transmitted adaptively with a $(6, 3)$ diagonal MDS code while $s_1$ and $s_2$ are transmitted nonadaptively with a lower rate $(4,1)$ code. Since $j = 0$, any erased source packet will necessarily have the corresponding symbols sent with a nonadaptive rate while any nonerased times will use an adaptive rate. Overall in this example, we then obtain a rate of 0.3, higher than the nonadaptive rate of 0.25 and slightly lower than the fully adaptive rate of 0.33. Meanwhile, the field size required is 6, which is lower than a requirement of 18 in the fully adaptive case.
  
  \begin{table*}[]
      \caption{Source-to-relay Code, $T=5, N_1=2, N_2=3, j=0$}
      \label{tab:srex1}
      \centering
      \resizebox{1.3\columnwidth}{!}{%
      \begin{tabular}{|l|l|l|l|l|l|l|l}
      \cline{1-7}
      Time $i$ & 0     & 1     & 2     & 3     & 4     & 5     \\ \cline{1-7}
      $s_i[1:3]$    & {\color[HTML]{3166FF} $s_0[1:3]$} & \cellcolor[HTML]{9B9B9B} {\color[HTML]{FE0000} $s_1[1:3]$} &  \cellcolor[HTML]{9B9B9B}{\color[HTML]{FFCB2F} $s_2[1:3]$} & {\color[HTML]{32CB00} $s_3[1:3]$} & $s_4[1:3]$ & $s_5[1:3]$  & 3    \\ \cmidrule[3pt]{1-7} 
        $s_{i-1}[1:3]$   &       & \cellcolor[HTML]{9B9B9B}{\color[HTML]{3166FF} $s_0[1:3]$} & \cellcolor[HTML]{9B9B9B}{\color[HTML]{FE0000} $s_1[1:3]$} &  {\color[HTML]{FFCB2F} $s_2[1:3]$} & {\color[HTML]{32CB00} $s_3[1:3]$} & $s_4[1:3]$  & 3    \\ \cline{1-7}
        $s_{i-2}[1:3]$   &       &  \cellcolor[HTML]{9B9B9B}     & \cellcolor[HTML]{9B9B9B}{\color[HTML]{3166FF} $s_0[1:3]$} & {\color[HTML]{FE0000} $s_1[1:3]$} &  {\color[HTML]{FFCB2F} $s_2[1:3]$} &  {\color[HTML]{32CB00} $s_3[1:3]$}  & 3 \\ \cline{1-7}
      \end{tabular}
      }
  \end{table*}
  
  \begin{table*}
      \caption{Relay-to-destination Code, $T=5, N_1=2, N_2=3, j=0$}
      \label{tab:rdex1}
      \centering
  \begin{tabular}{|l|l|l|l|l|l|l|l|l|l}
      \cline{1-9}
      Time $i$ & 0     & 1     & 2     & 3     & 4     & 5 &  6  & 7 \\ \cline{1-9}
              $s_{i-2}[1:3]$    & & &   & {\color[HTML]{FE0000} $s_1[1:3]$}  &  {\color[HTML]{FFCB2F} $s_2[1:3]$}  &   &      &  & 3   \\ \cmidrule[3pt]{1-9}
              $p_{i-3, 1}[1:3]$    & & &   &   & {\color[HTML]{FE0000}$s_1[1:3]$}  &  {\color[HTML]{FFCB2F} $s_2[1:3]$}  &      &  & 3  \\ \cline{1-9}
              $p_{i-4, 2}[1:3]$    & & &   &   &   &  {\color[HTML]{FE0000}$s_1[1:3]$} &    {\color[HTML]{FFCB2F} $s_2[1:3]$}   &  & 3   \\ \cline{1-9}
              $p_{i-5,3 }[1:3]$   & & &   &   &   &   &  {\color[HTML]{FE0000}$s_1[1:3]$}     &  {\color[HTML]{FFCB2F} $s_2[1:3]$} & 3   \\ \cmidrule[5pt]{1-9}
                      $s_{i}[1]$    & {\color[HTML]{3166FF} $s_0[1]$} &  &   &  {\color[HTML]{32CB00} $s_3[1]$}  & $s_4[1]$  & $s_5[1]$  &      &  & 1  \\  \cline{1-9}
              $s_{i-1}[2]$    & &{\color[HTML]{3166FF} $s_0[2]$} &   &    & {\color[HTML]{32CB00} $s_3[2]$}  & $s_4[2]$  &   $s_5[2]$   &  & 1  \\ \cline{1-9}
              $s_{i-2}[3]$    & & & {\color[HTML]{3166FF} $s_0[3]$}  &    &    & {\color[HTML]{32CB00} $s_3[3]$}  &   $s_4[3]$   & $s_5[3]$ & 1  \\ \cmidrule [3pt]{1-9}
              \rule{0pt}{0.5cm} $p_{i-3, 1}[1]$   & & &   &  \pbox{10cm}{\color[HTML]{3166FF} $s_0[1]$ \\ $ + s_0[2]$ \\ $+ s_0[3]$}  &    &  & \pbox{10cm}{\color[HTML]{32CB00} $s_3[1]$ \\ $ + s_3[2]$ \\ $+ s_3[3]$}     & \pbox{10cm}{  $s_4[1]$ \\ $ + s_4[2]$ \\ $+ s_4[3]$} & 1   \\ \cline{1-9}
              $p_{i-4, 2}[1]$    & & &   &   & \pbox{10cm}{\color[HTML]{3166FF} $s_0[1]$ \\ $ + 2s_0[2]$ \\ $+ 3s_0[3]$} &    &       & \pbox{10cm}{\color[HTML]{32CB00} $s_3[1]$ \\ $ + 2s_3[2]$ \\ $+ 3s_3[3]$} & 1  \\ \cline{1-9}
              $p_{i-5, 3}[1]$   & & &   &   &   & \pbox{10cm}{\color[HTML]{3166FF} $s_0[1]$ \\ $ + 3s_0[2]$ \\ $+ 4s_0[3]$}  &     &    & 1 \\ \cline{1-9}
          
      \end{tabular}
      \end{table*}

  As the second example, we consider the parameters $T = 6, N_1 = 2, N_2 = 3$ and $j = 1$, and assume that there are isolated erasures at $t = 4$ and $t=6$. We show the resulting source-to-relay and relay-to-destination code in Tables \ref{tab:srex2} and \ref{tab:rdex2} respectively. In this case, the nonerased packets are again sent adaptively, however due to the value of $j$ being 1, there is a delay of 1 timeslot. Meanwhile, $s_4$ is initially erased but then since the packet at time 5 is not erased, the relay has enough symbols to start transmission with the adaptive rate. Note that the relay observes erasures causally and thus cannot know beforehand whether it will be able to continue transmitting all of the symbols at this rate. 
  
 In this case, since there is an erasure at time 6 which leads to $2 > 1$ erasures affecting the packet from time 4, transmission cannot continue at the adaptive rate. Thus the remaining symbol (out of the 3 total) received by the relay at time 5 is transmitted at time 6 and from time 7 onwards the nonadaptive rate is used where the relay sends 3 symbols per timeslot. After time 7, all of the source symbols have been transmitted and so the packets will start to consist of parity symbols. Here, the concept of grouping is used, as we have described in Section \ref{sec:paritydesc}. In this example, this will mean that although the diagonal code used in the transmission had uneven packet sizes (2, 1, 3 at times 5, 6, 7 respectively), the encoding and decoding can be done assuming that 3 concatenations of a $(5, 2)$ code are used. Specifically, we will have $p_{4, 1}[1] = s_4[2] + s_4[1]$, $p_{4, 1} [2] = s_4[4] + s_4[3]$ and $p_{4, 1}[3] = s_4[6] + s_4[5]$. Similarly, $p_{4, 2}$ and $p_{4, 3}$ will continue with further linear combinations e.g. $p_{4,2}[1] = s_4[2] + 2s_4[1]$, $p_{4, 3} [1] = s_4[2] + 3 s_4[1]$. Then the destination will also decode symbols according to this MDS code. The first concatenation will recover symbols $s_4[2], s_4[1]$, the second will recover $s_4[4], s_4[3]$ and the third will recover $s_4[6],s_4[5]$. It can be verified that any three erasures in $[5:10]$ can cause at most three symbols in each individual $(5, 2)$ code to be erased. We show this process visually in Tables \ref{tab:ex2group1} and \ref{tab:ex2group2}. Table \ref{tab:ex2group1} shows the original transmission of packets by the relay and the two groups formed for encoding in green and orange. Table \ref{tab:ex2group2} shows the case where a $(5, 2)$ code was used from the beginning to illustrate how this grouping leads to the parity packets being constructed. As we can see, at most 3 symbols in each row (concatenation) of the code in Table \ref{tab:ex2group2} can be erased given any 3 packet erasures in the original transmission. We prove that this holds in general in Section \ref{sec:fieldsize}. We further note in this example that with the grouping we use 3 concatenations of a (5,2) MDS code instead of a (15, 6) long MDS code that would be used if the strategy for adaptation from prior work is used.

 Finally, the symbols from time 6 have only one erasure and so the transmission is again done using the adaptive rate. Therefore, to summarize this example we see that only the packet from time 4 is sent with a nonadaptive rate. This demonstrates intuitively why considering values of $j  > 0$ can be beneficial; while the adaptive rate is lower than the adaptive rate when $j = 0$, there can now be more source packets that use the adaptive rate. The value of the rate obtained in this example is 0.46, compared to a rate of 0.48 with the fully adaptive code and a rate of 0.4 in the nonadaptive case. Meanwhile, the field size required is 6, which is significantly lower than the field size of 96 required in the fully adaptive case.

    \begin{table*}[t]
      \caption{Source-to-relay code, $T =6, N_1 = 2, N_2 = 3, j=1$}
      \label{tab:srex2}
      \centering
      \resizebox{2\columnwidth}{!}{%
      \begin{tabular}{|l|l|l|l|l|l|l}
      \cline{1-6}
      Time $i$ & 3                                 & 4                                 & 5                                 & 6                                 & 7                                                                \\ \cline{1-6}
       $s_i[1:2:6]$    & \color[HTML]{3166FF} $s_3[1:2:6] $                   & \cellcolor[HTML]{9B9B9B}  \color[HTML]{FE0000} $s_4[1:2:6]  $              & \color[HTML]{FFCB2F} $s_5[1:2:6]  $       & \cellcolor[HTML]{9B9B9B} \color[HTML]{32CB00}  $s_6[1:2:6]  $         & $s_7[1:2:6]  $                   &    3 \\ \cline{1-6}
       $s_i[2:2:6]$   & \color[HTML]{3166FF} $s_3[2:2:6] $                   &\cellcolor[HTML]{9B9B9B}   \color[HTML]{FE0000} $s_4[2:2:6] $                    & \color[HTML]{FFCB2F} $s_5[2:2:6] $               & \cellcolor[HTML]{9B9B9B}  \color[HTML]{32CB00} $s_6[2:2:6] $      & $s_7[2:2:6]$   &   3 \\ \cline{1-6}
       \rule{0pt}{0.5cm} \pbox{10cm}{$s_{i-2}[1:2:6]$ \\ $+ s_{i-1}[2:2:6]$}    & \pbox{10cm}{$s_{1}[1:2:6]$ \\ $+ s_{2}[2:2:6]$}   & \cellcolor[HTML]{9B9B9B} \pbox{10cm}{$s_{2}[1:2:6]$ \\ $+$ \color[HTML]{3166FF}{$s_{3}[2:2:6]$}}   & \pbox{10cm}{\color[HTML]{3166FF}{$s_3[1:2:6]$} \\ \color[HTML]{000000}{$+$} \color[HTML]{FE0000}{$s_4[2:2:6]$}}       & \cellcolor[HTML]{9B9B9B}  \pbox{10cm}{\color[HTML]{FE0000}{$s_4[1:2:6]$} \\ \color[HTML]{000000}{$+$} \color[HTML]{FFCB2F}{$s_5[2:2:6]$}}  & \pbox{10cm}{\color[HTML]{FFCB2F}{$s_5[1:2:6]$} \\ \color[HTML]{000000}{$+$} \color[HTML]{32CB00}{$s_6[2:2:6]$}}   &   3 \\ \cline{1-6}
       \rule{0pt}{0.5cm} \pbox{10cm}{$s_{i-3}[1:2:6]$ \\ $+ 2s_{i-2}[2:2:6]$}    & \pbox{10cm}{$s_{0}[1:2:6]$ \\ $+ 2s_{1}[2:2:6]$} & \cellcolor[HTML]{9B9B9B}  \pbox{10cm}{$s_{1}[1:2:6]$ \\ $+ 2s_{2}[2:2:6]$} & \pbox{10cm}{$s_{2}[1:2:6]$ \\ $+$ \color[HTML]{3166FF}{$2s_{3}[2:2:6]$}}  & \cellcolor[HTML]{9B9B9B} \pbox{10cm}{\color[HTML]{3166FF}{$s_3[1:2:6]$} \\ \color[HTML]{000000}{$+$} \color[HTML]{FE0000}{$2s_4[2:2:6]$}} & \pbox{10cm}{\color[HTML]{FE0000}{$s_4[1:2:6]$} \\ \color[HTML]{000000}{$+$} \color[HTML]{FFCB2F}{$2s_5[2:2:6]$}} &    3 \\ \cline{1-6}
      \end{tabular}
      }
  \end{table*}

  \begin{table*}[t]
      \caption{Relay-to-destination code, $T =6, N_1 = 2, N_2 = 3, j=1$} 
      \label{tab:rdex2}
      \centering
      \resizebox{2\columnwidth}{!}{%
      \begin{tabular}{|l|l|l|l|l|l|l|l|l|l|l}
      \cline{1-10}
      Time $i$ & 3              & 4              & 5              & 6                             & 7              & 8              & 9              & 10              & 11              &   \\ \cline{1-10}
                                 & $s_2[1:3:6]$ &  \color[HTML]{3166FF} $s_3[1:3:6]$   &                & \color[HTML]{FFCB2F} $s_5[1:3:6]$               & \color[HTML]{32CB00} $s_6[1:3:6]$ &  & &  &  & 2 \\ \cline{1-10}
                                 & $s_1[2:3:6]$ & $s_2[2:3:6]$ &  \color[HTML]{3166FF} $s_3[2:3:6]$             &                               & \color[HTML]{FFCB2F} $s_5[2:3:6]$               & \color[HTML]{32CB00}   $s_6[2:3:6]$              &                &                 &                 & 2 \\ \cline{1-10}
                                 & $s_0[3:3:6]$ & $s_1[3:3:6]$ & $s_2[3:3:6]$ &    \color[HTML]{3166FF} $s_3[3:3:6]$                &           &   \color[HTML]{FFCB2F} $s_5[3:3:6]$             &   \color[HTML]{32CB00}   $s_6[3:3:6]$            &                &                 &                  2 \\ \cline{1-10}
                           \rule{0pt}{0.5cm}   &                & \pbox{10cm}{$s_{0}[1:3:6]$ \\ $+ s_0[2:3:6]$ \\ $+s_0[3:3:6]$}     & \pbox{10cm}{$s_{1}[1:3:6]$ \\ $+ s_1[2:3:6]$ \\ $+s_1[3:3:6]$}      & \pbox{10cm}{$s_{2}[1:3:6]$ \\ $+ s_2[2:3:6]$ \\ $+s_2[3:3:6]$}                & \pbox{10cm}{\color[HTML]{3166FF} $s_{3}[1:3:6]$ \\ $+ s_3[2:3:6]$ \\ $+s_3[3:3:6]$}    &                &  \pbox{10cm}{\color[HTML]{FFCB2F} $s_{5}[1:3:6]$ \\ $+ s_5[2:3:6]$ \\ $+s_5[3:3:6]$}       &  \pbox{10cm}{\color[HTML]{32CB00} $s_{6}[1:3:6]$ \\ $+ s_6[2:3:6]$ \\ $+s_6[3:3:6]$}      &                 & 2 \\ \cline{1-10}
                           \rule{0pt}{0.5cm}   &            &    & \pbox{10cm}{$s_{0}[1:3:6]$ \\ $+2 s_0[2:3:6]$ \\ $+3s_0[3:3:6]$}     & \pbox{10cm}{$s_{1}[1:3:6]$ \\ $+ 2s_1[2:3:6]$ \\ $+3s_1[3:3:6]$}      & \pbox{10cm}{$s_{2}[1:3:6]$ \\ $+ 2s_2[2:3:6]$ \\ $+3s_2[3:3:6]$}                & \pbox{10cm}{\color[HTML]{3166FF} $s_{3}[1:3:6]$ \\ $+ 2s_3[2:3:6]$ \\ $+3s_3[3:3:6]$}    &                &  \pbox{10cm}{\color[HTML]{FFCB2F} $s_{5}[1:3:6]$ \\ $+2s_5[2:3:6]$ \\ $+3s_5[3:3:6]$}       &  \pbox{10cm}{\color[HTML]{32CB00} $s_{6}[1:3:6]$ \\ $+ 2s_6[2:3:6]$ \\ $+3s_6[3:3:6]$}      &              2 \\ \cline{1-10}      
                           \rule{0pt}{0.5cm}   &        &    &    & \pbox{10cm}{$s_{0}[1:3:6]$ \\ $+3 s_0[2:3:6]$ \\ $+4s_0[3:3:6]$}     & \pbox{10cm}{$s_{1}[1:3:6]$ \\ $+ 3s_1[2:3:6]$ \\ $+4s_1[3:3:6]$}      & \pbox{10cm}{$s_{2}[1:3:6]$ \\ $+ 3s_2[2:3:6]$ \\ $+4s_2[3:3:6]$}                & \pbox{10cm}{\color[HTML]{3166FF} $s_{3}[1:3:6]$ \\ $+ 3s_3[2:3:6]$ \\ $+4s_3[3:3:6]$}    &                &  \pbox{10cm}{\color[HTML]{FFCB2F} $s_{5}[1:3:6]$ \\ $+3s_5[2:3:6]$ \\ $+4s_5[3:3:6]$}       &       2 \\ \cline{1-10}
                                 &                &                & \color[HTML]{FE0000} $s_4[2:2:4]$              &                 &  &        &         &         &                 & 2 \\ \cline{1-10}
                                 &                &                &             &  \color[HTML]{FE0000} $s_4[6]$            &   &         &         &        &                 & 1 \\ \cline{1-10}
                                 \rule{0pt}{0.5cm}    &                &                &               &                 &  \color[HTML]{FE0000} $s_4[1:2:6]$ & \pbox{10cm}{ \color[HTML]{FE0000} $s_4[2] + s_4[1]$ \\ $\frown s_4[4] + s_4[3]$ \\  $\frown s_4[6] + s_4[5]$}    &  \pbox{10cm}{ \color[HTML]{FE0000} $s_4[2] + 2s_4[1]$ \\ $\frown s_4[4] + 2s_4[3]$ \\  $\frown s_4[6] + 2s_4[5]$}    &  \pbox{10cm}{ \color[HTML]{FE0000} $s_4[2] + 3s_4[1]$ \\ $\frown s_4[4] + 3s_4[3]$ \\  $\frown s_4[6] +3 s_4[5]$}   &                 & 3 \\ \cline{1-10}
      \end{tabular}
      }
      \end{table*}

            \begin{table*}[]
        \caption{Original Transmission for symbols from time 4}
        \label{tab:ex2group1}
        \centering
        \resizebox{\columnwidth}{!}{%
        \begin{tabular}{|l|l|l|l|l|l|}
          \hline
          5                                & 6                                & 7                                & 8            & 9            & 10           \\ \hline
          \cellcolor[HTML]{67FD9A}$s_4[2]$ & \cellcolor[HTML]{67FD9A}$s_4[6]$ & \cellcolor[HTML]{FE996B}$s_4[1]$ & $p_{4,1}[1] $ & $p_{4,2}[1]$ & $p_{4,3}[1]$ \\ \hline
          \cellcolor[HTML]{67FD9A}$s_4[4]$ &                                  & \cellcolor[HTML]{FE996B}$s_4[3]$ & $p_{4,1}[2]$ & $p_{4,2}[2]$ & $p_{4,3}[2]$ \\ \hline
                                           &                                  & \cellcolor[HTML]{FE996B}$s_4[5]$ & $p_{4,1}[3]$ & $p_{4,2}[3]$ & $p_{4,3}[3]$ \\ \hline
          \end{tabular}
        }
        \end{table*}

        \begin{table*}[]
          \caption{Transmission with grouping of $(5,2)$ code}
          \label{tab:ex2group2}
          \centering
          \resizebox{\columnwidth}{!}{%
          \begin{tabular}{|l|l|l|l|l|}
          \hline
          6                                & 7                                & 8            & 9            & 10           \\ \hline
          \rule{0pt}{0.3cm}  \cellcolor[HTML]{67FD9A}$s_4[2]$ & \cellcolor[HTML]{FE996B}$s_4[1]$ & \pbox{10cm}{$p_{4,1}[1] =$ \\ $ s_4[2] + s_4[1]$}   & \pbox{10cm}{$p_{4,2}[1] =$ \\ $ s_4[2] + 2s_4[1]$}  & \pbox{10cm}{$p_{4,3}[1] =$ \\ $ s_4[2] + 3s_4[1]$} \\ \hline
          \rule{0pt}{0.3cm} \cellcolor[HTML]{67FD9A}$s_4[4]$ & \cellcolor[HTML]{FE996B}$s_4[3]$ & \pbox{10cm}{$p_{4,1}[2] =$ \\ $ s_4[4] + s_4[3]$}   & \pbox{10cm}{$p_{4,2}[2] =$ \\ $ s_4[4] + 2s_4[3]$}  & \pbox{10cm}{$p_{4,3}[2] =$ \\ $ s_4[4] + 3s_4[3]$} \\ \hline
          \cellcolor[HTML]{67FD9A}$s_4[6]$ & \cellcolor[HTML]{FE996B}$s_4[5]$ & \pbox{10cm}{$p_{4,1}[3] =$ \\ $ s_4[6] + s_4[5]$}   & \pbox{10cm}{$p_{4,2}[3] =$ \\ $ s_4[6] + 2s_4[5]$}  & \pbox{10cm}{$p_{4,3}[3] =$ \\ $ s_4[6] + 3s_4[5]$} \\ \hline
          \end{tabular}
          }
          \end{table*}

    \section{Analysis}\label{sec:anal}

    In this section, we present a complete proof for our main result, Theorem \ref{thm:rates}. The main novelty is in the method and corresponding analysis to obtain the said field size requirements (\ref{sec:fieldsize}) which then leads to a trivial proof of $(N_1, N_2)$-achievability of the scheme (\ref{sec:achievable}).
 
  \subsection{Worst Case Length of Relay Packets}\label{sec:worstcase}
  
  We calculate a general expression for the relay packet size $n_2$ when 
  the erasure patterns we adapt to are the ones with less than or equal to $j$ erasures. 
  Specifically, we are interested in the worst case value so that the relay can deal with any valid erasure pattern.
  Let us analyze the packet size of an arbitrary packet $x_t^{(2)}$. Note that, in general, information from packets $\{s_{t'}\}_{t' = t - T}^{t}$ will contribute to packet $x_t^{(2)}$. However, since we only adapt to at least $j$ erasures, the latest packet that may contribute to $x_t^{(2)}$ is $s_{t - j}$, rather than $s_{t}$. We now wish to analyze how much each packet $\{s_{t'}\}_{t' = t - T}^{t - j}$ contributes to the packet size of $x_{t}^{(2)}$. We assume that $i \leq N_1$ erasures happen at the time slots $\{ \tau_1, \dots, \tau_{i} \} \subseteq [t-T:t]$ (here we still consider the full time window from $t-T$ to $t$ as there can be at most $N_1$ erasures in a window of size $T+1$ by Remark \ref{rem:1} in the system model). 
  By Algorithm \ref{alg:subset}, we know that $\alpha_t(t+i) \leq \ell_{\gamma(i)}$. Since $i-i'$ erasures happen between time $\tau_{i'} + 1$ and time $t$, we can then obtain $\alpha_{\tau_i'} (t - \tau_i') \leq \ell_{i - i'}$, $i' \in [1:i]$. Recall that $\alpha_{\tau_i'}(t-\tau_i')$ represents the contribution of packet $s_{\tau_{i'}}$ to the relay packet $x_t^{(2)}$.

  We can then bound the packet length of $x_t^{(2)}$ as 
  
  \begin{equation}
      \tilde{n} \leq \underbrace{(T+1-i-j) \cdot \frac{k}{T+1-N_2-j}}_\text{contribution of non-erased packets} + \underbrace{\sum_{i'' \in [0:i-1]} \ell_{i''}}_\text{contribution of erased packets}.
  \end{equation}
  
  Here $T+1-j$ is the total length of the time window $[t-T:t-j]$ so there are at least $(T+1-i-j)$ non-erased packets
  , each of which is transmitted at a rate of $\frac{T+1-N_2-j}{T+1-j}$ with $l'' = \frac{k}{T+1-N_2-j}$ symbols per timeslot. Note that for the second term, we have let $i'' = i-i'$, 
  where this quantity represents the number of erasures in the time window $[\tau_{i'}+1, t]$ i.e., the number of erasures that affect packet $s_{\tau_{i'}}$, which then defines the number of symbols that are transmitted from that packet at time $t$. 
  We now simplify the second term as 
  
  \begin{align}
      \begin{split}
      \sum_{i'' \in [0: \min(i-1, j-1)]} \frac{k}{T+1-N_2-j} \\ + \sum_{i'' \in [\min(i, j):i-1]} \frac{k}{T+1-N_2-N_1}
      \end{split} \\
      &= \frac{k \cdot (\min(i, j))}{T +1 - N_2 - j } + \frac{k \cdot (i - \min(i, j))}{T +1 - N_2 - N_1}
  \end{align}
  by the definition of $\ell_{\gamma(i)}$. Specifically, for $i'' \leq j-1$, $\ell_{i''} = \frac{k}{T+1-N_2-j}$ and for $i'' > j$, $\ell_{i''} = \frac{k}{T+1-N_2-N_1}$.

  Combining the two equations, we then obtain 
  
  \begin{equation}
      \tilde{n} \leq \frac{k \cdot (T+ 1 - \max(i, j))}{T+1-N_2-j} + \frac{k \cdot (i - \min(i, j))}{T+1-N_2-N_1}
  \end{equation}
  
  Since $N_1 > j$, $\frac{1}{T+1-N_2-N_1} > \frac{1}{T+1-N_2-j}$ and so this expression is maximized by the highest possible value of $i$. This is $i = N_1$ by the problem definition (note that $T+1-N_2-j \geq 0$ is satisfied since $T+1-N_2-N_1 \geq 0$ and $j < N_1$). Therefore, we obtain 
  
  \begin{equation}\label{eq:n2worst}
      \begin{split}
      & n_2^* = \frac{k \cdot (T+ 1 - N_1)}{T+1-N_2-j} + \frac{k \cdot (N_1 - j)}{T+1-N_2-N_1} \\ &= (T+1-N_2-N_1)(T+ 1 - N_1) \\ &+ (T+1-N_2-j)(N_1 - j)
      \end{split}
  \end{equation}
  
  The rate $R_2 \triangleq \frac{k}{n_2}$ can thus be expressed as

  \begin{equation}\label{eq:r2}
      \begin{split}
      &R_2 = \\ &\frac{(T+1-N_2-j)(T+1-N_2-N_1)}{(T+1-N_1)(T+1-N_2-N_1) + (N_1 - j)(T+1-N_2-j)}
      \end{split}
  \end{equation}
  
  , matching the expression in Theorem \ref{thm:rates} for $R_2$. The value of $R_1$ follows directly from the source-to-relay encoding discussed earlier. 
  
  To choose the optimal value of $j$, we perform a grid search over all values of $j$ from 0 to $N_1$ and choose the one which gives the highest value of $R_2$. 
  Finally, an example of this upper bound on $n_2$ being achieved with equality can be seen in Table \ref{tab:srex1} where the relay packet at time 5 has a length of 10, equal to Equation \ref{eq:n2worst} with $T = 6, N_1 = 2, N_2 = 3, j = 0$. 

  \subsection{Field Size Requirements}\label{sec:fieldsize}

  In this section, we prove that a field size of $q = T+1-j$ is sufficient. Recall that depending on the erasures observed, we mentioned using either a $(T+1-j, T+1-j-N_2)$ or a $(T+1-N_1, T+1-N_1-N_2)$ MDS code from relay to destination, concatenated several times. Meanwhile, the source-to-relay code is always fixed as a $(T+1-N_2, T+1-N_1-N_2)$ MDS code concatenated $T+1-N_2-j$ times. Using a property of MDS codes, we know that with $q \geq n$ we can construct an MDS code over $\mathbb{F}_q$. Out of all the values of $n'$ for the 3 codes mentioned, we see that $T+1-j$ is the maximum and thus with $q = T+1-j$, we can construct MDS codes. It remains to show that the claimed codes are indeed robust against any $N_2$ erasures, especially for the case where $x_t^{(1)}$ is erased. 

  As discussed before, if $x_t^{(1)}$ is not erased or if $j = 0$, the relay-to-destination code is fixed and so applying these MDS codes is straightforward as the sizes of each subpacket $C(t; i)$ are constant. For the case of $j > 0$, where the sizes of the subpackets may change, we detailed in Section \ref{sec:subsetrdenc} how a notion of grouping is used to allow short MDS codes to continue being used. 
  
  While it is not necessary that all symbols in a particular group $g_{r, t}$ are sent as part of the same relay packet initially (e.g. when higher rates are used or there are a low number of available symbols), we now claim that the parity symbols are sufficient to recover all of the $g_{r, t}$'s. For this, we first consider the received packets at the destination, $y_{t+j}^{(2)}$ to $y_{t+T}^{(2)}$. $y_{t+j}^{(2)}$ to $y_{t+T-N_2}^{(2)}$ include the message symbol estimates from time $t$ while $y_{t+T-N_2+1}^{(2)}$ to $y_{t+T}^{(2)}$ include the parity symbols for the diagonal codes. The number of message symbols from time $t$ contained in one packet can be between 0 and $T+1-j-N_2$ while the number of parity symbols is always equal to $T+1-j-N_2$.
  
  We now focus specifically on the packets containing the message symbol estimates i.e. $\dpack{t+j}$ to $\dpack{t+T-N_2}$. We note that each such packet can contain symbols from at most two groups as the size of each relay subpacket (the $C(t; i)$'s) is always less than or equal to $T+1-j-N_2$, and the relay transmits the $g_i$'s in order. In total, we thus have $T+1-N_1-N_2$ groups distributed over $T+1-j-N_2$ time slots. Now when one of these relay packets is erased, some portion of a group or two groups are erased with a total of $\leq T+1-j-N_2$ symbols erased. We then present the following proposition.
  
  \begin{proposition}\label{prop:group}
    Let there be a total of $w \leq N_2$ erasures among the packets $\dpack{t+j}$ to $\dpack{t+T-N_2}$. Then at most $w$ of the symbols $g_{1, t}[p], \dots, g_{T+1-N_1-N_2, t}[p]$ are erased for all $p$ with $1 \leq p \leq T+1-j-N_2$
  \end{proposition}
  
  \begin{proof}
    
  We proceed by contradiction and assume that there is a $p$ for which $w + 1$ of these symbols are erased. Let one of the relay packet erasures cause $g_{q, t}[p]$ be erased. Then we note that this erasure cannot lead to another of the symbols $g_{1, t}[p], \dots, g_{T+1-N_1-N_2, t}[p]$ being erased since in order of transmission, consecutive symbols here have $T+1-j-N_2$ symbols in between ($g_{i, t}[p], \dots, g_{i, t}[T+1-j-N_2], g_{i+1, t}[1], \dots, g_{i+1, t}[p-1]$ are a total of $T+1-j-N_2$ symbols) and relay subpackets have a size less than or equal to $T+1-j-N_2$. Therefore, having $w + 1$ erased symbols for index $p$ implies that there are $w + 1$ erasures among the packets $\dpack{t+j}$ to $\dpack{t+T-N_2}$. However, this is a contradiction as we assumed that there were $w$ erasures among these packets.
  
  \end{proof}

 From our initial assumption about there being $w$ erasures among packets $\dpack{t+j}$ to $\dpack{t+T-N_2}$, we know that there can be at most $N_2 - w$ erasures among the packets $\dpack{t+T-N_2+1}$ to $\dpack{t+T}$. Then using Proposition \ref{prop:group}, we know that the $(T+1-N_1, T+1-N_1-N_2)$ MDS code corresponding to each value of $p$ has at most $w + N_2 - w = N_2$ erasures. Hence by the properties of MDS codes, we know that this code will be able to correct any $N_2$ erasures and thus the destination can recover the symbols $g_{1, t}[p], \dots, g_{T+1-N_1-N_2, t}[p]$ for all values of $p$. In other words, all source symbol estimates will be decoded at the destination. We note here that the destination uses a decoding function $\varphi_t$ corresponding to a $(T+1-N_1, T+1-N_1-N_2)$ code.

  \subsection{Recoverability of $s_t$ at destination by time $t+T$}\label{sec:achievable}
  
  We start by recalling the following proposition from \cite{adaptiverelay2023} and prove it for our setting. 
  
  \begin{proposition}
      Using this coding scheme, if there are at most $N_2$ erasures from relay to destination, the destination is able to recover an estimate $\tilde{s}_t$ at time $t+T$ as well as $s_t$. 
  \end{proposition}

  \begin{proof}
      We have two main cases for the source to relay packets. If $x_t^{(1)}$ is not erased, all of $s_t$ is received and thus the only code used from relay to destination will be a $(T+1-j, T+1-N_2-j)$-MDS code concatenated $l''$ times which by definition can correct any $n'' - k'' = N_2$ erasures.
  
      Similarly, if $\spack{t}$ is erased, the code from relay to destination will either be a $(T+1-j, T+1-N_2-j)$ MDS code if there are less than or equal to $j$ erasures, or a $(T+1-N_1, T+1-N_2-N_1)$ MDS code if there are more than $j$ erasures. Note that this also covers the possibility that within message variable rates are used as discussed in Section \ref{sec:fieldsize}. For all of these possibilities, we can see that the MDS codes have $n'' - k'' = N_2$ and can thus correct any $N_2$ erasures. 

      Finally, given that we have shown in all cases that the symbol estimates $\tilde{s}_t$ can be recovered for all $t$ by time $t+T$, an induction argument can be used to show that the original symbols are also recovered at time $t+T$. This is because we can remove the interference in symbol estimates by using the values of the symbols from previous times.
  
  \end{proof}
  
  This proposition shows that our scheme is $(N_1, N_2)$-achievable, thereby allowing for the proof of Theorem \ref{thm:rates} to be completed. Finally, for completeness the relay can inform the destination of the erasure pattern it observes. At any time $t$, it can forward the observed erasure sequence from time $t-T$ to $t$. Doing so requires sending a binary sequence of length $T+1$ as a header which can be represented by $\delta = \lceil (T+1) \log_q 2 \rceil$ symbols. 

\section{Results}\label{sec:numresult}

  In this section, we analyze the performance of our proposed scheme in terms of achievable rates and packet size requirements for various parameter settings $T, N_1, N_2$ as well as the loss probability when simulated over different types of statistical channels.

  \subsection{Performance}
  
   Figures \ref{fig:rate} and \ref{fig:length} compare the achievable rate and packet size requirements of our scheme against the nonadaptive SWDF scheme (\cite{silas2019}) and the standard adaptive relaying scheme (\cite{adaptiverelay2023}) over various choices of system parameters.

  \begin{figure}[htbp]
      \centerline{\includegraphics[width=\linewidth]{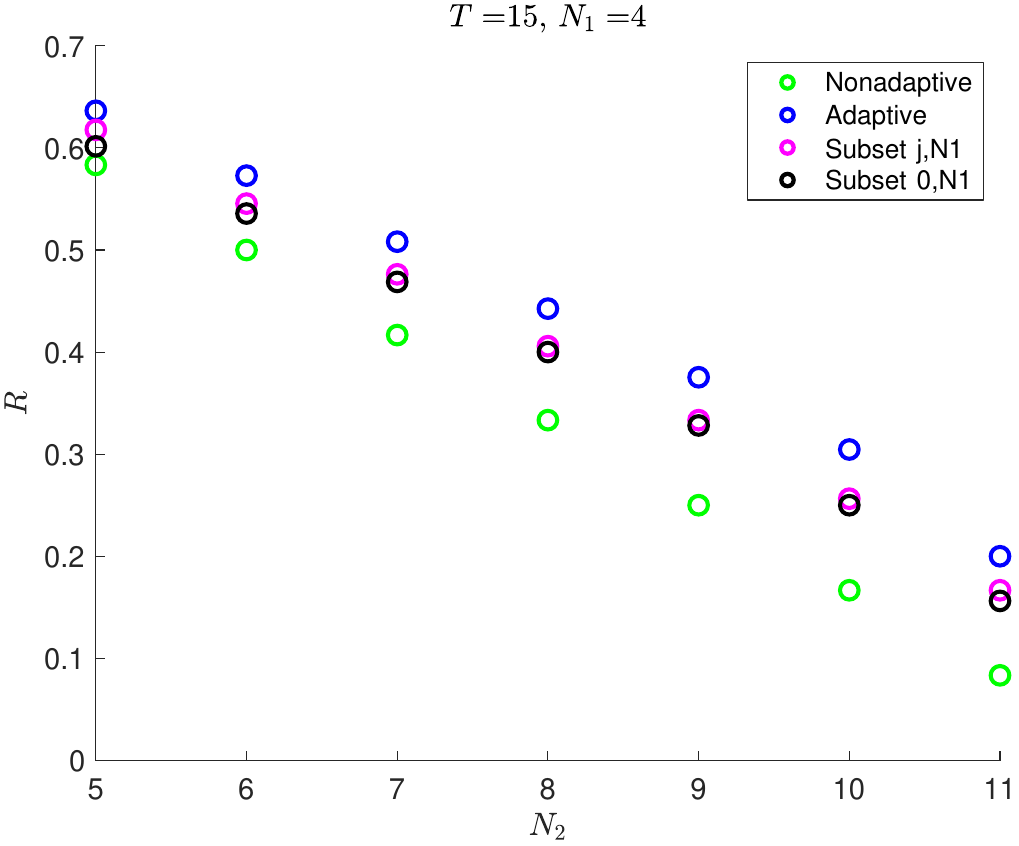}}
      \caption{Rates of each method with fixed values of $T$ and $N_1$, varying $N_2$. }
      \label{fig:rate}
      \end{figure}

  \begin{figure}[htbp]
      \centerline{\includegraphics[width=\linewidth]{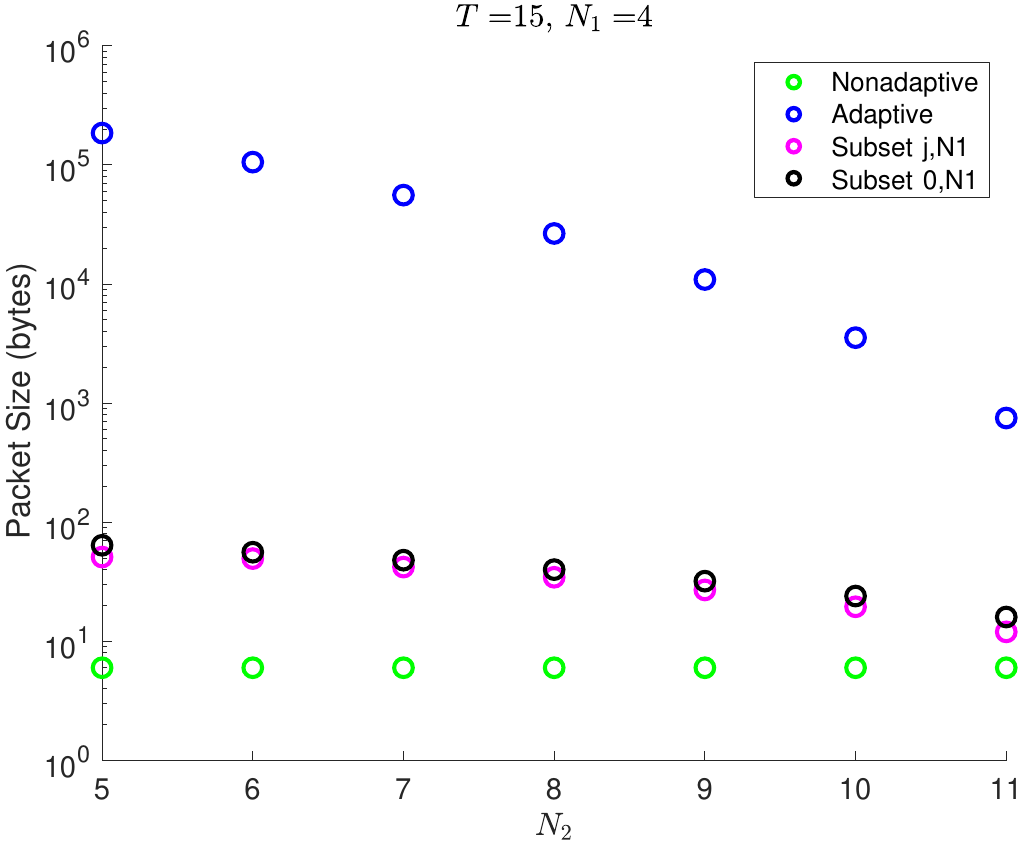}}
      \caption{Packet size of each method with fixed values of $T$ and $N_1$, varying $N_2$.}
      \label{fig:length}
      \end{figure}

  For our method, we show both the case of fixing $j = 0$ (Subset $0, N_1$) and using the optimal value of $j$ (Subset $j, N_1$). As we can see, while the rate of our method is upper bounded by the standard adaptive scheme, it comes with the benefit of having several orders of magnitude lower packet sizes. We also note that while the optimal values of $j$ may vary between parameters, in practice this change is not significant and for most combinations the optimal value is $j = 1$. 

  The packet size calculation follows directly from the field size (Theorem \ref{thm:rates}) and the number of symbols used in each packet. As an example computation of the packet sizes (in bytes) for all three methods, we consider the parameters $T=15, N_1=4, N_2=6$. For standard adaptive codes, a long MDS code is used with $n = 52788$. This means that we need $q \geq 52788$ and so every symbol will need $\lceil \log_2 52788 \rceil = 16$ bits or 2 bytes to represent. The size of the overall packet is thus $52788 \times 2 \approx 106$ KB which is significantly larger than packet sizes in typical internet protocols such as TCP where packet sizes of 1-2 Kb are used. Meanwhile for subset adaptation, we use short MDS codes with $q = T+1 - j = 16$. Thus each symbol needs $\log_2 16 = 4$ bits to represent and the size of the overall packet is then $4 \times n_2 = 4 \times 112 = 448$ bits or 56 bytes. Finally, the nonadaptive scheme only uses one concatenation of a short MDS code with $n = 12$ and this requires $\lceil \log_2 12 \rceil \cdot 12 = 48$ bits or 6 bytes to represent.

  \subsection{Numerical Simulations}
  
  We perform simulations over a statistical channel to evaluate the performance of our scheme against random erasure sequences. Here an erasure happens at any time in either of the two links with i.i.d probabilities of $\alpha, \beta$ respectively. Since it is now possible that the number of erasures within a window of size $T$ exceeds $N_1, N_2$, some packets may be lost. We thus compute the probability of a packet being lost and compare it with the nonadaptive scheme from \cite{silas2019}. Note that computing the corresponding loss probabilities for the general adaptive relaying method \cite{adaptiverelay2023} is intractable and so we do not compare against it. This is another advantage of our scheme as it can be analyzed more precisely under different types of erasure channels. The method to compute the loss probability for the non adaptive case has been described in \cite{silas2019}. Specifically, a given packet $s_t $ is classified as being lost if there are more than $N_1$ erasures in the S-R link during the times over which the diagonals containing symbols in $s_t$ are transmitted i.e. at least one of the symbols in $s_t$ is not decoded by the relay, or if there are more than $N_2$ erasures in the R-D link for the times over which the symbols are then forwarded by the relay.

  We now describe the method to compute the loss probability for our scheme. First, we let $j = 0$. Then if $x_t^{(1)}$ is not erased, the relay will have all symbols belonging to $s_t$ at time $t$ and therefore will transmit these symbols at a rate of $\frac{k}{T+1-N_2}$ i.e. over the interval $[t, t+T-N_2]$ after which the $N_2$ timeslots $[t+T-N_2+1, t+T]$ will be used for parity symbols. In other words, we can see that if there are more than $N_2$ erasures in the interval $[t, t+T]$, the destination will not be able to decode all symbols from $s_t$ and this is thus the condition for a packet loss. If on the other hand $x_t^{(1)}$ was erased, then the relay begins transmitting at time $t+N_1$ and we have the same scenario as the nonadaptive case. 
  
  For $j > 0$, if the relay has observed less than or equal to $j$ erasures by time $t+j$, it will begin transmitting at $t + j$. In this case, the interval where symbols are transmitted is $[t+j, t+T]$ and so if there are more than $N_2$ erasures in $[t+j, t+T]$, a packet loss occurs. Here, since symbol estimates are used, we also need to check for erasures in the first link to ensure that each diagonal has less than or equal to $N_1$ erasures so that all of the previous symbols used in the estimate of $s_t$ can be decoded and there is no propagation of error. 
  Finally if the relay has observed greater than $j$ erasures at time $t+j$, it will instead transmit starting at $t+N_1$ as in the nonadaptive case and then the loss probability calculation is the same as before.

 For our simulations, we use 10000000 total packets and vary the erasure probabilites $\alpha, \beta$. For each value of $i$, we determine if $x_i$ is recovered successfully at the destination and calculate the loss probability as the fraction of lost packets to the total number of packets. We present results where the probability of erasure $\alpha=\beta$ is varied under a fixed rate 
    in Figure \ref{fig:lossprob1} and the results where the rate is varied (by changing the values of $N_1, N_2$) under fixed values of erasure probabilities $\alpha, \beta$ in Figure \ref{fig:lossprob2}. 
  
  \begin{figure}[htbp]
  \centerline{\includegraphics[width=\linewidth]{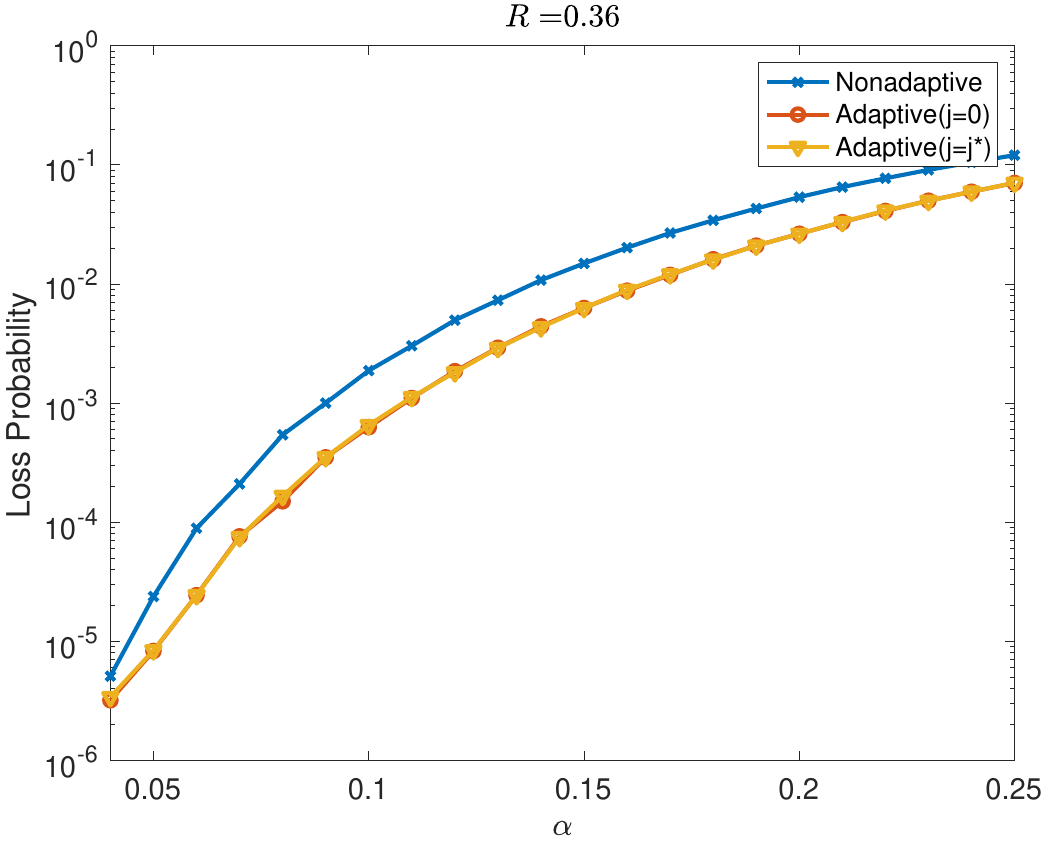}}
  \caption{Loss probability of all three approaches with a fixed rate of 0.36 and varying channel erasure probabilities $\alpha = \beta$.} 
  \label{fig:lossprob1}
  \end{figure}

  \begin{figure}[htbp]
  \centerline{\includegraphics[width=\linewidth]{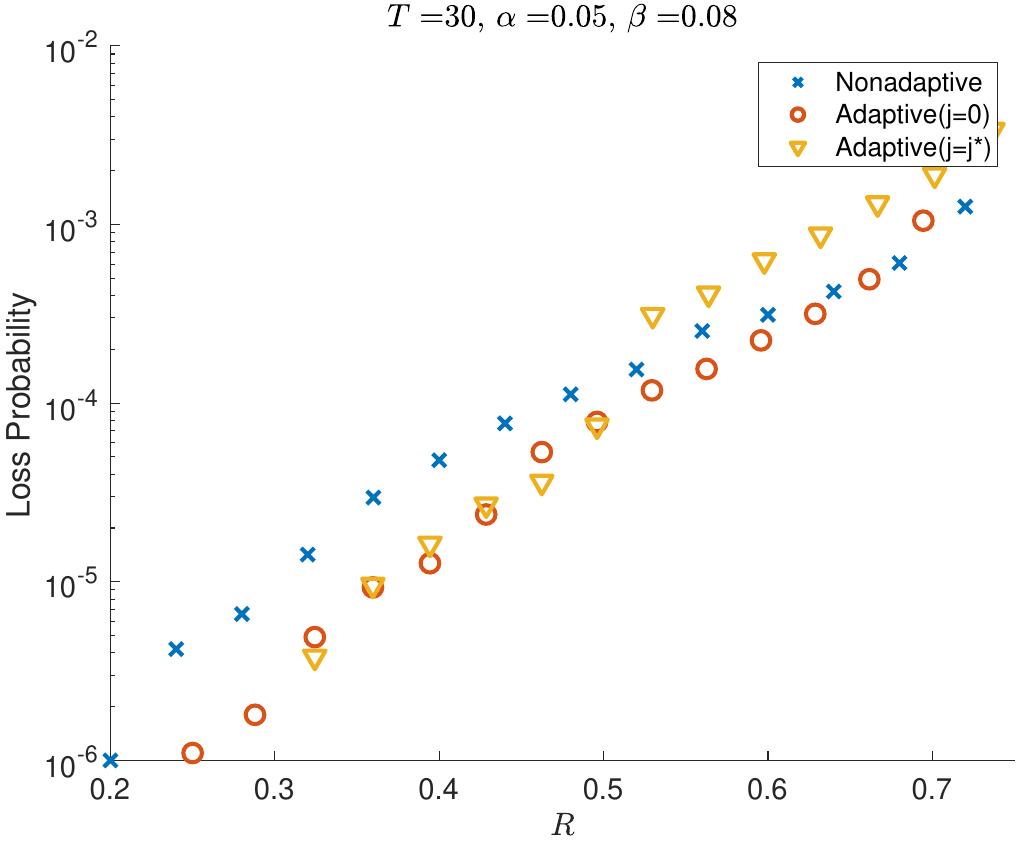}}
  \caption{Loss probability of all three approaches with varying rate and channel erasure probability $\alpha = 0.05$, $\beta = 0.08$.}
  \label{fig:lossprob2}
  \end{figure}
  
  From these plots, we see in general that at lower rates, the adaptive approach gives a lower loss probability than the nonadaptive case whereas at higher rates, the loss probabilities come closer together. Intuitively this is the case because a lower rate corresponds to a higher value of $N_2$. Since the condition for a loss due to the first link is the same for all methods, we focus on the losses caused by the second link. As in \cite{adaptiverelay2023}, we also only focus on the case where $N_2 \geq N_1$ so that the adaptive methods have a higher rate for the same parameters. With a higher value of $N_2$, the codes can withstand more erasures in the second link without suffering from a packet loss and thus for fixed probabilities of erasure will have a better rate vs loss probability tradeoff. At higher rates (lower values of $N_2$), the codes now have a larger probability of erasure in the second link and so there will be more patterns which cause a loss for a fixed value of $N_2$. The condition for the adaptive method is that there are $> N_2$ erasures in the interval $[t, t+T]$ while for the nonadaptive case the interval is $[t+N_1, t+T]$ (since symbols from $s_t$ are only transmitted starting at $t, t+N_1$ respectively). With the interval being larger in the adaptive case, there is thus a greater probability of loss there. To summarize, there is a tradeoff between the improved rate offered by the adaptive methods which generally leads to lower error probabilities for low values of $R$ but as $R \to 1$, the probability for the nonadaptive scheme approaches or becomes lower.

\section{Extensions to Multiaccess Adaptation}\label{sec:extension}

In this section, we show how it is possible to apply the ideas from our scheme to a more complicated multiple access network setting \cite{kasper2022}. Here, the network is again modeled as a relay network, with the difference being that there are two (or more) sources, as shown in Figure \ref{fig:MARC}. The streaming code is furthermore described as a $(n_1, n_2, n_3, k_1, k_2, T)$ streaming code. Each of the three links is unreliable and introduces at most $N_1, N_2, N_3$ erasures respectively within a window of size $T+1$. Corresponding to the two sources, there are now two rates $R_1$ and $R_2$ with $R_i = \frac{k_i}{\max(n_1, n_2, n_3)}$ and these two rates together form a rate pair. The enclosure of all achievable rate pairs forms a rate region. \cite{kasper2022} has studied this setting for nonadaptive codes and has derived upper bounds for the network, with bounds on the individual values of $R_1, R_2$ as well as on the sumrate $R_1 + R_2$. These bounds are as in the following equations. 

\begin{align}
  R_1 & \leq C(T-N_3, N_1) \\
  R_2 & \leq C(T-N_3, N_2) \\
  R_1 + R_2 & \leq C(T-N_2, N_3) 
\end{align}

\begin{figure*}
  \centering
  \includegraphics[draft=false,scale=1]{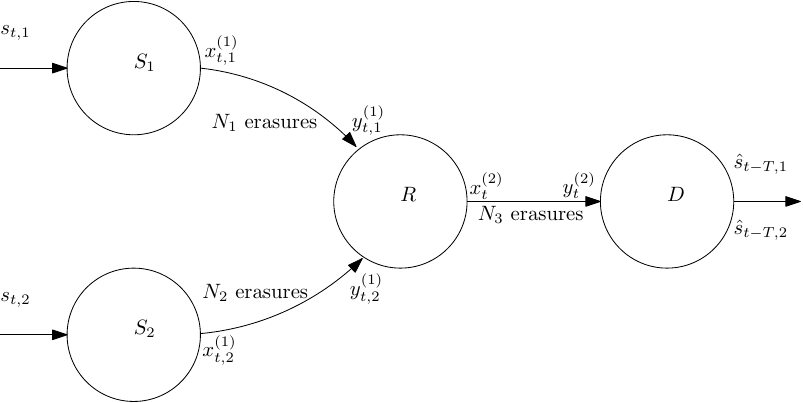}
  \caption{Multiple Access Relay Channel \cite{kasper2022}}
  \label{fig:MARC}
\end{figure*}

The single user bounds hold for all possible codes, however the sumrate bound is dependent on the codes being nonadaptive. The natural question then is if rate pairs higher than the sumrate can be achieved using adaptation.

We find that this is indeed possible with a trivial extension of our scheme. Specifically, we consider two copies of the subset adaptive code, one for the system parameters being $(T, N_1, N_3, j_1)$ and the other for the parameters $(T, N_2, N_3, j_2)$. The discussion in Section \ref{sec:subset} then leads to the formation of $(n_1, n_3, k_1, T)$ and $(n_2, n_3, k_2, T)$ single user streaming codes respectively. These two codes can then be converted to multiaccess streaming codes by letting $n_2 = k_2 = 0$ in the former and $n_1 = k_1 = 0$ in the latter. From this, we have $(N_1, N_2, N_3)$ achievable $(n_1, 0, n_3, k_1, 0, T)$ and $(0, n_2, n_3, 0, k_2, T)$ streaming codes. Finally, the following lemma from \cite{kasper2022} is used to concatenate the two codes and construct a rate region.

\begin{lemma}\label{lemma:timesharing}
	Assume there exists an $(N_1, N_2, N_3)$-achievable $(n_1, n_2, n_3, k_1, k_2, T)_{\mathbb{F}}$-streaming code, and another $(N_1, N_2, N_3)$-achievable $(n'_1, n'_2, n'_3, k'_1, k'_2, T)_{\mathbb{F}}$-streaming code. Then, for any $A, B \in \mathbb{Z}$, there exists an $(N_1, N_2, N_3)$-achievable $(A n_1 + B n_1', An_2 + B n_2', A n_3 + B n_3', A k_1 + B k_1', A k_2 + B k_2', T)_{\mathbb{F}}$-streaming code.
\end{lemma}

\begin{proof}
    In \cite{kasper2022}.
\end{proof}

We present the result of a rate region achieved by this adaptive scheme in Figure \ref{fig:macrate}. We can see that we obtain a rate region higher than the sumrate at almost all points, thus showing significant gains over nonadaptive schemes. Moreover, the resulting scheme also benefits from the practicality and field size improvements obtained by our subset scheme as we have only concatenated multiple copies of single user subset codes so that the underlying MDS codes (and thus the field size) still remain the same. Specifically, the field size requirement will be $\max(T+1-j_1, T+1-j_2)$ where $j_1, j_2$ are the values of $j$ used by the two users.

  \begin{figure}[htbp]
    \centerline{\includegraphics[width=\linewidth]{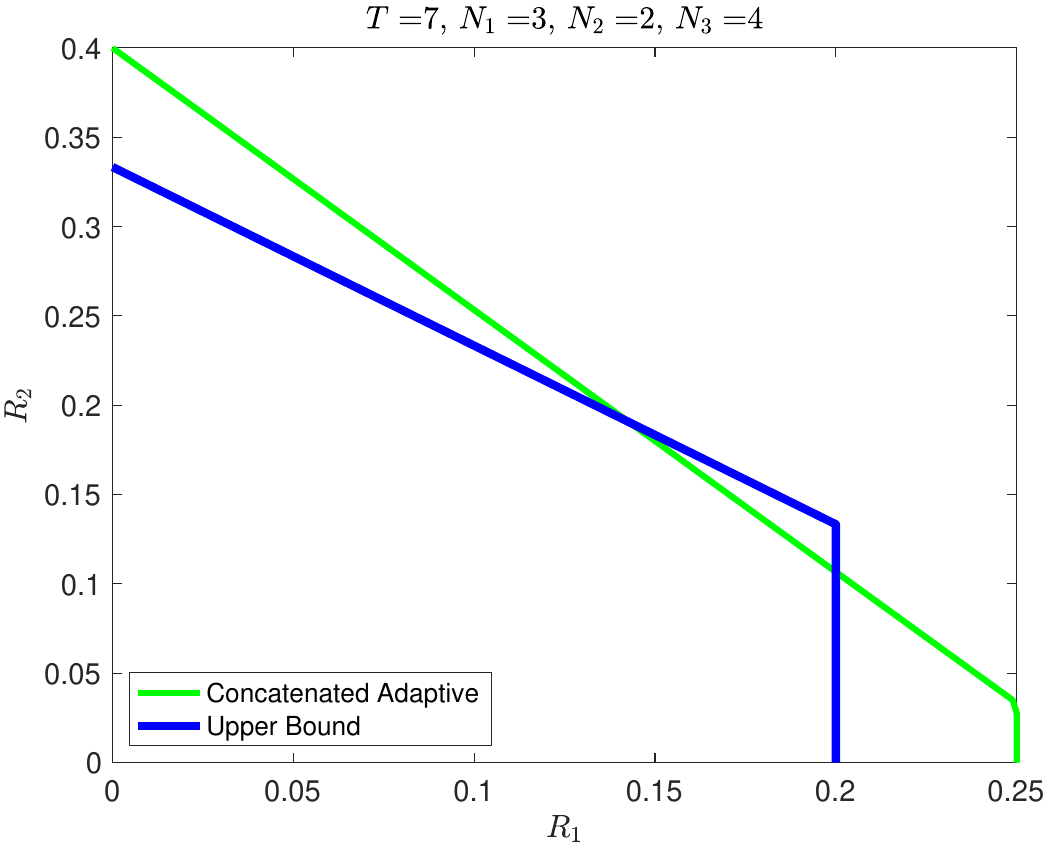}}
    \caption{Adaptive rate region, $T=7, N_1=3,N_2=2,N_3=4$.}
    \label{fig:macrate}
    \end{figure}

\section{Limitations and Future Work}

While this paper has focused on the case where only a single number of erasures is adapted to, specifically $j$, it is possible to consider a larger number of possibilities and thus more options for adaptive rates. Note that including all numbers from 0 to $N_1$ corresponds to the fully adaptive case as in \cite{adaptiverelay2023} amd thus we do not consider it. Using a larger number of possibilities will allow for higher rates compared to our approach but again will have the tradeoff of causing the packet size to increase. Future work can characterize the full tradeoff which exists when allowing more possibilites for the adaptation power of the relay. Other areas of interest include designing adaptive codes for specific types of channel erasure models or adaptive codes which allow for variable delay constraints for different packets. Similarly, one can consider adaptive code constructions with different levels of error protection which may achieve higher rates at the cost of not having a formal guarantee that all packets can be decoded successfully.

\section{Conclusion}

In this paper, we have presented a practical adaptive relaying scheme for streaming codes in a three node relay network. Building on the idea of adaptation, our simplified scheme allows for significantly improved practicality with a small cost to the achievable rates. We have compared our scheme against prior proposed schemes in terms of achievable rates, packet size requirements and loss probabilities, showing the utility of our scheme in achieving a good tradeoff between practicality and improved rates. Finally, we have discussed an extension to the multiaccess network setting where the advantages of our scheme can be directly applied.

\bibliographystyle{IEEEtran}
\bibliography{IEEEabrv,paper}

\begin{thebibliography}{10}
\providecommand{\url}[1]{#1}
\csname url@samestyle\endcsname
\providecommand{\newblock}{\relax}
\providecommand{\bibinfo}[2]{#2}
\providecommand{\BIBentrySTDinterwordspacing}{\spaceskip=0pt\relax}
\providecommand{\BIBentryALTinterwordstretchfactor}{4}
\providecommand{\BIBentryALTinterwordspacing}{\spaceskip=\fontdimen2\font plus
\BIBentryALTinterwordstretchfactor\fontdimen3\font minus
  \fontdimen4\font\relax}
\providecommand{\BIBforeignlanguage}[2]{{%
\expandafter\ifx\csname l@#1\endcsname\relax
\typeout{** WARNING: IEEEtran.bst: No hyphenation pattern has been}%
\typeout{** loaded for the language `#1'. Using the pattern for}%
\typeout{** the default language instead.}%
\else
\language=\csname l@#1\endcsname
\fi
#2}}
\providecommand{\BIBdecl}{\relax}
\BIBdecl

\bibitem{silas2019}
S.~L. {Fong}, A.~{Khisti}, B.~{Li}, W.~{Tan}, X.~{Zhu}, and
  J.~{Apostolopoulos}, ``Optimal streaming erasure codes over the three-node
  relay network,'' in \emph{2019 IEEE International Symposium on Information
  Theory (ISIT)}, July 2019, pp. 3077--3081.

\bibitem{adaptiverelay2023}
G.~K. Facenda, M.~Nikhil~Krishnan, E.~Domanovitz, S.~L. Fong, A.~Khisti, W.-T.
  Tan, and J.~Apostolopoulos, ``Adaptive relaying for streaming erasure codes
  in a three node relay network,'' \emph{IEEE Transactions on Information
  Theory}, pp. 1--1, 2023.

\bibitem{kasper2022}
G.~K. Facenda, E.~Domanovitz, A.~Khisti, W.-T. Tan, and J.~Apostolopoulos,
  ``Streaming erasure codes over multi-access relayed networks,'' \emph{IEEE
  Transactions on Information Theory}, vol.~69, no.~2, pp. 860--885, 2023.

\bibitem{martinian2004burst}
E.~Martinian and C.-E. Sundberg, ``Burst erasure correction codes with low
  decoding delay,'' \emph{IEEE Transactions on Information theory}, vol.~50,
  no.~10, pp. 2494--2502, 2004.

\bibitem{leong2012erasure}
D.~Leong and T.~Ho, ``Erasure coding for real-time streaming,'' in \emph{2012
  IEEE International Symposium on Information Theory Proceedings}, 2012, pp.
  289--293.

\bibitem{badr2013streaming}
A.~Badr, A.~Khisti, W.-T. Tan, and J.~Apostolopoulos, ``Streaming codes for
  channels with burst and isolated erasures,'' in \emph{2013 Proceedings IEEE
  INFOCOM}.\hskip 1em plus 0.5em minus 0.4em\relax IEEE, 2013, pp. 2850--2858.

\bibitem{joshi2012playback}
G.~Joshi, Y.~Kochman, and G.~W. Wornell, ``On playback delay in streaming
  communication,'' in \emph{2012 IEEE International Symposium on Information
  Theory Proceedings}.\hskip 1em plus 0.5em minus 0.4em\relax IEEE, 2012, pp.
  2856--2860.

\bibitem{Karzand2017}
M.~{Karzand}, D.~J. {Leith}, J.~{Cloud}, and M.~{Médard}, ``Design of {FEC}
  for low delay in {5G},'' \emph{IEEE Journal on Selected Areas in
  Communications}, vol.~35, no.~8, pp. 1783--1793, 2017.

\bibitem{badr2017layered}
A.~Badr, P.~Patil, A.~Khisti, W.-T. Tan, and J.~Apostolopoulos, ``Layered
  constructions for low-delay streaming codes,'' \emph{IEEE Transactions on
  Information Theory}, vol.~63, no.~1, pp. 111--141, 2017.

\bibitem{badr2017fec}
A.~Badr, A.~Khisti, W.-t. Tan, X.~Zhu, and J.~Apostolopoulos, ``{FEC} for
  {V}o{IP} using dual-delay streaming codes,'' in \emph{IEEE INFOCOM 2017-IEEE
  Conference on Computer Communications}.\hskip 1em plus 0.5em minus
  0.4em\relax IEEE, 2017, pp. 1--9.

\bibitem{krishnan2018rate}
M.~N. Krishnan and P.~V. Kumar, ``Rate-optimal streaming codes for channels
  with burst and isolated erasures,'' in \emph{2018 IEEE International
  Symposium on Information Theory (ISIT)}.\hskip 1em plus 0.5em minus
  0.4em\relax IEEE, 2018, pp. 1809--1813.

\bibitem{fong2019optimal}
S.~L. {Fong}, A.~{Khisti}, B.~{Li}, W.~{Tan}, X.~{Zhu}, and
  J.~{Apostolopoulos}, ``Optimal streaming codes for channels with burst and
  arbitrary erasures,'' \emph{IEEE Transactions on Information Theory},
  vol.~65, no.~7, pp. 4274--4292, 2019.

\bibitem{domanovitz2019explicit}
E.~Domanovitz, S.~L. Fong, and A.~Khisti, ``An explicit rate-optimal streaming
  code for channels with burst and arbitrary erasures,'' \emph{arXiv preprint
  arXiv:1904.06212}, 2019.

\bibitem{KrishnanLowField2020}
M.~N. {Krishnan}, D.~{Shukla}, and P.~V. {Kumar}, ``Low field-size,
  rate-optimal streaming codes for channels with burst and random erasures,''
  \emph{IEEE Transactions on Information Theory}, pp. 1--1, 2020.

\bibitem{Haghifam2021}
M.~Haghifam, M.~N. Krishnan, A.~Khisti, X.~Zhu, W.-T. Tan, and
  J.~Apostolopoulos, ``On streaming codes with unequal error protection,''
  \emph{IEEE Journal on Selected Areas in Information Theory}, vol.~2, no.~4,
  pp. 1165--1179, 2021.

\bibitem{Jarschel2013}
\BIBentryALTinterwordspacing
M.~Jarschel, D.~Schlosser, S.~Scheuring, and T.~Hoßfeld, ``Gaming in the
  clouds: Qoe and the users’ perspective,'' \emph{Mathematical and Computer
  Modelling}, vol.~57, no.~11, pp. 2883--2894, 2013, information System
  Security and Performance Modeling and Simulation for Future Mobile Networks.
  [Online]. Available:
  \url{https://www.sciencedirect.com/science/article/pii/S0895717711007771}
\BIBentrySTDinterwordspacing

\bibitem{Quax2013}
P.~Quax, A.~Beznosyk, W.~Vanmontfort, R.~Marx, and W.~Lamotte, ``An evaluation
  of the impact of game genre on user experience in cloud gaming,'' in
  \emph{2013 IEEE International Games Innovation Conference (IGIC)}, 2013, pp.
  216--221.

\bibitem{Clincy2013}
V.~Clincy and B.~Wilgor, ``Subjective evaluation of latency and packet loss in
  a cloud-based game,'' in \emph{2013 10th International Conference on
  Information Technology: New Generations}, 2013, pp. 473--476.

\bibitem{Claypool2014}
M.~Claypool and D.~Finkel, ``The effects of latency on player performance in
  cloud-based games,'' in \emph{2014 13th Annual Workshop on Network and
  Systems Support for Games}, 2014, pp. 1--6.

\bibitem{Slivar2014}
I.~Slivar, M.~Suznjevic, L.~Skorin-Kapov, and M.~Matijasevic, ``Empirical qoe
  study of in-home streaming of online games,'' in \emph{2014 13th Annual
  Workshop on Network and Systems Support for Games}, 2014, pp. 1--6.

\bibitem{lee2014outatime}
\BIBentryALTinterwordspacing
K.~Lee, D.~Chu, E.~Cuervo, J.~Kopf, S.~Grizan, A.~Wolman, and J.~Flinn,
  ``Outatime: Using speculation to enable low-latency continuous interaction
  for cloud gaming,'' Tech. Rep. MSR-TR-2014-115, August 2014. [Online].
  Available:
  \url{https://www.microsoft.com/en-us/research/publication/outatime-using-speculation-to-enable-low-latency-continuous-interaction-for-cloud-gaming/}
\BIBentrySTDinterwordspacing

\bibitem{Slivar2015}
I.~Slivar, M.~Suznjevic, and L.~Skorin-Kapov, ``The impact of video encoding
  parameters and game type on qoe for cloud gaming: A case study using the
  steam platform,'' in \emph{2015 Seventh International Workshop on Quality of
  Multimedia Experience (QoMEX)}, 2015, pp. 1--6.

\end{thebibliography}

\end{document}